\documentclass[pdflatex,sn-mathphys-num]{sn-jnl}

\usepackage{amsfonts,amssymb,mathtools}
\usepackage{array,booktabs,graphicx}
\usepackage{enumitem}
\hypersetup{
  pdftitle={Multi-View Block Distance Distributions and Linear Programming Bounds for Locally Recoverable Codes with Availability},
  pdfauthor={Ming-Hsuan Kang, Maosheng Xiong, Yu Hsuan Hsieh, and Po-Wei Lai},
  pdfkeywords={locally recoverable codes, availability, multi-view block distance distributions, four-block distance distributions, linear programming bounds, Krawtchouk polynomials}
}

\theoremstyle{thmstyleone}
\newtheorem{theorem}{Theorem}[section]
\newtheorem{proposition}[theorem]{Proposition}
\newtheorem{lemma}[theorem]{Lemma}

\theoremstyle{thmstylethree}
\newtheorem{definition}[theorem]{Definition}
\theoremstyle{thmstyletwo}
\newtheorem{remark}[theorem]{Remark}

\newcommand{\C}{\mathcal C}
\newcommand{\F}{\mathbb F}
\newcommand{\Kraw}{K}
\newcommand{\1}{\mathbf 1}
\newcommand{\Cond}{\mathrm{Cond}}
\newcommand{\FB}{\mathrm{4blk}}
\numberwithin{table}{section}

\begin{document}

\title[Multi-view block LP bounds for LRCs]
{Multi-View Block Distance Distributions and Linear Programming Bounds for
Locally Recoverable Codes with Availability}

\author{\fnm{Ming-Hsuan} \sur{Kang}}\nomail
\author{\fnm{Maosheng} \sur{Xiong}}\nomail
\author{\fnm{Yu Hsuan} \sur{Hsieh}}\nomail
\author{\fnm{Po-Wei} \sur{Lai}}\nomail

\abstract{We develop a multi-view linear-programming framework for locally
recoverable codes with arbitrary fixed availability \(a\). For any retained
order \(1\le s\le a\), the selected helper sets, the recovered coordinate,
and their complement form an \(s+2\)-part partition. Recording the Hamming
distance on all blocks preserves both compatibility among the selected repair
alternatives and their coupling with the remaining coordinates. The resulting
joint distribution satisfies centered counting identities, product-Krawtchouk
positivity, and collision inequalities from the local-distance condition; for
fixed \(s\), these constraints give a polynomial-size relaxation for
arbitrary, possibly nonlinear, codes over any finite field. Retaining one
view recovers the three-block model of our companion paper. We develop the
first genuinely multi-view case, \(s=2\), in detail and specialize the exact
computations to availability \(a=2\). The resulting four-block LP projects
to both the one-view three-block model and a globally conditioned two-view
relaxation, making explicit the information lost by each coarsening. Exact
rational primal--dual certificates together with checked constructions prove
\(M_{\max}(2,8,4,2,2,2)=8\),
\(M_{\max}(3,7,3,2,2,2)=27\), and
\(M_{\max}(4,7,3,2,2,2)=64\); in all three cases the four-block bound is
strictly stronger than both coarsenings.}



\maketitle

\section{Introduction}
\label{sec:introduction}

Locally recoverable codes (LRCs) reduce repair cost by reconstructing an
erased coordinate from a small local view \cite{gopalan2012,prakash2012}.
In the classical availability model, a coordinate has \((r,t)\)-availability
if it can be reconstructed from \(t\) pairwise disjoint helper sets, each
of size at most \(r\) \cite{rawat2016,tamobargfrolov2016}. This model has
led to rate and distance bounds, shortening and recovery-graph arguments,
parity-check and generalized-weight methods, and many explicit constructions
\cite{wangzhangliu2015,balajikumar2017,kruglik2017,kruglik2019}.

The code class studied here combines availability with the local-distance
form of locality introduced in \cite{prakash2012}. Each selected view
\(T_{i,h}=\{i\}\mathbin{\dot\cup}S_{i,h}\) has length at most
\(r+\delta-1\) and the projected code \(\C|_{T_{i,h}}\) has minimum
distance at least \(\delta\). Hence every selected view is itself a local
code that corrects up to \(\delta-1\) erasures. This hypothesis implies
functional availability but is generally stronger even when \(\delta=2\):
standard availability only requires the distinguished coordinate to be
determined by its helpers, whereas local distance two requires the entire
local projection to correct one erasure. Multiple available views protected
by local codes of prescribed distance have also been studied, for example in
the information-symbol setting of Kadiyam and Das \cite{kadiyamdas2020}.

Our earlier work introduced the three-block distance distribution attached
to a single recovery view \cite{paperone2026}. The present paper places that
construction in a multi-view hierarchy. For a code with availability \(a\),
one may retain any ordered number \(s\) of disjoint repair views with
\(1\le s\le a\). The corresponding \(s+2\)-block statistic records the
Hamming distance on the \(s\) helper sets, the recovered coordinate, and
their complement. Taking \(s=a\) retains all available repair views, while
\(s=1\) is exactly the three-block model of \cite{paperone2026}. For fixed
\(s\), centered counting, coordinate-partition Fourier positivity, and local
collision arguments give a polynomial-size relaxation; before symmetry
reduction its local state space has size at most \(O(nR_0^{2s})\), where
\(R_0=\min\{n,r+\delta-1\}\).

We develop the first genuinely multi-view case, \(s=2\), in detail and
specialize the exact computations to availability \(a=2\). This case already
contains the two interactions that the method is designed to retain. Treating
two repair views independently forgets
whether their local distance patterns arise from the same ordered pair of
codewords. Recording only their union preserves that compatibility but
forgets how each local Fourier sector extends across the coordinates outside
the two views. The resulting four-block distribution retains both pieces of
information while keeping the notation and exact computation manageable.

Fix a recovered coordinate \(i\), choose two disjoint helper sets
\(S_{i,1}\) and \(S_{i,2}\), and let \(O_i\) be the complement of the
selected union. Then
\[
 S_{i,1},\qquad S_{i,2},\qquad \{i\},\qquad O_i
\]
partition \([n]\). For each ordered triple \((x,y,i)\), we record the
Hamming distance on each block. The resulting distribution has three
layers. First, centered counting identities recover the global distance
distribution. Second, the product-Hamming character transform is
nonnegative in all four spectral degrees; the positive degrees on \(O_i\)
record how the two local views extend globally. Third, local distance gives
collision inequalities for each local view and for the punctured helper union.
Together these constraints define the four-block LP.

The comparison with coarser bounds is built into the same object. Merging
one helper block with \(O_i\) returns exactly to the one-view three-block LP
introduced in our earlier work \cite{paperone2026}. Keeping both helper
distances and the total distance,
but discarding positive spectral degrees on \(O_i\), gives a conditioned
two-view relaxation. These are projections of the four-block distribution,
not parallel theories. The comparison theorem therefore identifies which
information each coarsening discards.

The main contributions are as follows.
\begin{enumerate}[label=(\roman*)]
\item We formulate the multi-view hierarchy through a retained order
\(1\le s\le a\). The resulting \(s+2\)-block statistic contains the
three-block model of \cite{paperone2026} at \(s=1\) and remains
polynomial-size for every fixed \(s\).
\item We develop \(s=2\) in detail. For the resulting four-block
distribution we derive the centered identities, full product-Krawtchouk
positivity, and local collision inequalities, and prove objective-preserving
projections to both the one-view three-block LP and the conditioned two-view
relaxation. The formulation uses \(O(nr^4)\) local variables for fixed
\(\delta\).
\item Exact rational primal--dual certificates and independently checked
constructions determine three availability-two maxima over
\(\F_2,\F_3,\F_4\), and in each case the four-block optimum is strictly
smaller than both coarsened optima.
\end{enumerate}

This framework is complementary to existing availability bounds. Classical
availability work primarily controls rate, dimension, or minimum distance by
shortening, recovery graphs, parity-check incidence, or generalized Hamming
weights \cite{tamobargfrolov2016,balajikumar2017,kruglik2019}. Here the target is
an alphabet-dependent finite-length bound on
\(M_{\max}(q,n,d,r,\delta,a)\), valid for arbitrary, possibly nonlinear,
codes satisfying the stronger local-distance availability condition. The LP
retains the joint block-distance pattern of several repair alternatives for
the same ordered codeword pair and recovered coordinate, together with its
coupling to the coordinates outside the selected views.

Among locality-aware LP bounds, Agarwal et al. use a product association
scheme for a partition into uniform disjoint repair groups
\cite{agarwal2018}. The refined dual-distribution LP of Gruica, Jany, and
Ravagnani is support-sensitive and linear \cite{gruica2026}, while the
moment-based LP of Li, Wei, and Xiong uses higher local moments and applies
to general linear and nonlinear \(r,\delta\)-LRCs
\cite{liweixiong2026}. These methods encode locality through different state
spaces from the same-coordinate multi-view distribution studied here.
Maximally recoverable availability codes impose a prescribed topology and a
stronger erasure-correction requirement \cite{martinezpenas2026}; they form
a different code class.

Section~\ref{sec:availability-baseline} fixes the code class and the global
Delsarte layer. Section~\ref{sec:multiblock} develops the retained-order \(s=2\)
case, namely the four-block distribution, its positivity and collision
constraints, and the projection hierarchy. Section~\ref{sec:exact-bounds} gives the certified exact bounds
and the broader numerical landscape. Section~\ref{sec:limitations}
concludes with the scope of the relaxation.

\section{Code class and preliminaries}
\label{sec:availability-baseline}

Let \(q\) be a prime power and let \(\C\subseteq\F_q^n\) be a code of size
\(M\ge2\) and minimum Hamming distance at least \(d\). Linearity is not
assumed. Write \([m]=\{1,\ldots,m\}\), let \(\C|_E\) denote projection on
\(E\subseteq[n]\), and write \(d_E(x,y)\) for restricted Hamming distance.
Throughout, \(r\ge1\), \(\delta\ge2\), and \(a\ge1\). A projection with
fewer than two distinct words is assigned minimum distance zero.

\begin{definition}[all-symbol local-distance availability]
\label{def:availability}
The code \(\C\) has all-symbol local-distance
\((r,\delta,a)\)-availability if, for every \(i\in[n]\), one can select
ordered local views
\[
 T_{i,h}=\{i\}\mathbin{\dot\cup}S_{i,h},
 \qquad h=1,\ldots,a,
\]
such that
\[
 |T_{i,h}|\le r+\delta-1,\qquad
 d(\C|_{T_{i,h}})\ge\delta,
\]
and \(S_{i,1},\ldots,S_{i,a}\subseteq[n]\setminus\{i\}\) are pairwise
disjoint for fixed \(i\).
\end{definition}

No disjointness is required between recovery sets attached to different
recovered coordinates. The parameter \(\delta\) here is the minimum
distance of each local projection, following the \(r,\delta\)-locality
convention of \cite{prakash2012}. This should not be confused with the
notation \( (r,\delta)_c\) of Wang and Zhang \cite{wangzhang2014}, where
\(\delta-1\) counts alternative repair options rather than the distance of a
local code.

\begin{remark}[retained order and the general multi-view construction]
\label{rem:general-availability}
Fix \(1\le s\le a\) and retain the first \(s\) selected views. With
\[
 O_i^{(s)}=[n]\setminus
 \Bigl(\{i\}\mathbin{\dot\cup}S_{i,1}\mathbin{\dot\cup}\cdots
 \mathbin{\dot\cup}S_{i,s}\Bigr),
\]
the natural joint statistic is
\[
 \bigl(d_{S_{i,1}}(x,y),\ldots,d_{S_{i,s}}(x,y),
 \1_{\{x_i\ne y_i\}},d_{O_i^{(s)}}(x,y)\bigr).
\]
Its distribution is the \(s+2\)-block distance distribution. The centered
identities and product-Krawtchouk positivity follow blockwise exactly as
below. Each retained local view gives a complete-view collision inequality.
For a union of at least two helper sets, if the recovered coordinates agree,
any changed helper block already contributes at least \(\delta\) differences;
if they disagree, every selected helper set contributes at least
\(\delta-1\), so the punctured union has distance at least
\(2(\delta-1)\ge\delta\). Thus every such subset gives the corresponding
punctured helper-union inequality. Taking \(s=a\) retains all available
repair views, while \(s=1\) is the three-block model of
\cite{paperone2026}. We write out only \(s=2\) below; the exact extremal
computations further specialize to \(a=2\).
\end{remark}

For the detailed four-block model we now assume \(a\ge2\), fix one
admissible ordered selection, and retain its first two views. Put
\[
 t_h=|T_{i,h}|,\qquad
 R_0=\min\{n,r+\delta-1\},\qquad h=1,2.
\]
Then
\[
 \delta\le t_h\le R_0,\qquad t_1+t_2-1\le n.
\]
Let \(M_{\max}(q,n,d,r,\delta,a)\) denote the largest size of a code in
Definition~\ref{def:availability}. Since every code with \(a\ge2\) has two
selected views, the four-block
relaxation below is valid uniformly for all \(a\ge2\). We use it as the
complete retained-order \(s=2\) model; the exact extremal results in
Section~\ref{sec:exact-bounds} specialize to \(a=2\).

The local condition forces every nonzero codeword difference to have weight
at least \(\delta\): choose a coordinate \(i\) on which the pair differs and
use \(d(\C|_{T_{i,h}})\ge\delta\). We therefore use the effective threshold
\[
 d_\star=\max\{d,\delta\}.
\]

Define the normalized global distance distribution
\[
 A_\ell=\frac1M
 |\{(x,y)\in\C^2:d_H(x,y)=\ell\}|,
 \qquad 0\le\ell\le n.
\]
Thus
\[
 A_0=1,\qquad \sum_{\ell=0}^n A_\ell=M,\qquad
 A_\ell=0\quad(1\le\ell<d_\star).
\]
For \(0\le p,j\le N\), let
\[
 \Kraw_p(j;N,q)=
 \sum_{u=0}^p(-1)^u(q-1)^{p-u}
 \binom ju\binom{N-j}{p-u}
\]
be the \(q\)-ary Krawtchouk polynomial. The Delsarte character-energy
identity gives
\[
 B_w=\sum_{\ell=0}^n A_\ell\Kraw_w(\ell;n,q)\ge0,
 \qquad 0\le w\le n
\]
\cite{delsarte1973}. These mass, distance, and transform constraints form
the global layer of the four-block LP.

For comparison, let \(U_{\mathrm{3blk}}\) denote the optimum of the known
one-view three-block LP \cite{paperone2026}, which retains the partition
\(S_i\mathbin{\dot\cup}\{i\}\mathbin{\dot\cup}O_i\). We do not repeat
that formulation. The second comparison bound will arise directly as a
coarsening of the four-block distribution in Section~\ref{sec:multiblock}.

\section{The retained-order \texorpdfstring{$s=2$}{s=2} case: the four-block linear program}
\label{sec:multiblock}

The four-block construction is designed to keep two kinds of information at
once: compatibility between the two local views for the same ordered pair,
and the extension of that local pattern to the coordinates outside the
selected views. The LP is obtained in three steps. We first record the exact
four-block distances and their centered relation to the global distance
distribution. We then impose the full product-Krawtchouk spectrum, including
the outside spectral degree. Finally, local distance yields collision
inequalities on the two local views and their punctured helper union. The
resulting LP admits two natural projections, making the information
discarded by each comparison bound explicit.

Fix the first two selected local views at each recovered coordinate \(i\).
Write
\[
 O_i=[n]\setminus
 \bigl(S_{i,1}\mathbin{\dot\cup}S_{i,2}\mathbin{\dot\cup}\{i\}\bigr),
 \qquad
 u(t_1,t_2)=t_1+t_2-1.
\]
Thus
\begin{equation}
 [n]=S_{i,1}\mathbin{\dot\cup}S_{i,2}
 \mathbin{\dot\cup}\{i\}\mathbin{\dot\cup}O_i.
\label{eq:multiblock-partition}
\end{equation}
The admissible length pairs form
\[
 \mathcal T_2=
 \{(t_1,t_2):\delta\le t_1,t_2\le R_0,\
                  u(t_1,t_2)\le n\}.
\]

\subsection{The four-block distribution and centered identities}

For \((t_1,t_2)\in\mathcal T_2\),
\(0\le j<t_1\), \(0\le k<t_2\),
\(0\le c\le n-u(t_1,t_2)\), and \(b\in\{0,1\}\), define
\begin{equation}
 Z_{t_1,t_2;j,k,c,b}
 =\frac1{nM}\left|\left\{(x,y,i):
 \begin{array}{l}
 |T_{i,1}|=t_1,\ |T_{i,2}|=t_2,\\
 d_{S_{i,1}}(x,y)=j,\quad d_{S_{i,2}}(x,y)=k,\\
 d_{O_i}(x,y)=c,\quad \1_{\{x_i\ne y_i\}}=b
 \end{array}\right\}\right|.
\label{eq:Z-definition}
\end{equation}
Every cell records the global distance intrinsically:
\begin{equation}
 \ell=j+k+c+b.
\label{eq:Z-global-distance}
\end{equation}
Thus global-distance conditioning is already built into the four-block
array; no separate conditioned array is needed in the main formulation.

Counting, at each global distance \(\ell\), the recovered coordinates on
which an ordered pair agrees or differs gives
\begin{align}
 \sum_{\substack{t_1,t_2,j,k,c\\j+k+c=\ell}}
 Z_{t_1,t_2;j,k,c,0}
 &=\frac{n-\ell}{n}A_\ell,
\label{eq:Z-centered-zero}\\
 \sum_{\substack{t_1,t_2,j,k,c\\j+k+c+1=\ell}}
 Z_{t_1,t_2;j,k,c,1}
 &=\frac{\ell}{n}A_\ell.
\label{eq:Z-centered-one}
\end{align}
Local and global minimum distance imply
\begin{align}
 Z_{t_1,t_2;j,k,c,b}&=0
 &&\text{if }1\le j+b<\delta
       \text{ or }1\le k+b<\delta,
\label{eq:Z-view-zero}\\
 Z_{t_1,t_2;j,k,c,b}&=0
 &&\text{if }1\le j+k+c+b<d_\star.
\label{eq:Z-global-zero}
\end{align}
No additional zero constraint on the selected union is needed. Indeed,
\(1\le j+k+b<\delta\) already forces one of the two local zero conditions in
\eqref{eq:Z-view-zero}. Thus the four-block array contains the global
distance conditioning without introducing a second conditioned variable.

\subsection{Product-Krawtchouk positivity}

The Fourier constraint is a direct blockwise form of the Delsarte
character-energy identity. We record it because this formulation makes two
features used later transparent: it applies to nonlinear codes, and every
block carries its own spectral degree. Fix a nontrivial additive character
\(\chi\) of \(\F_q\) and write
\(\chi_\alpha(x)=\chi(\alpha\mathbin{\cdot}x)\).

\begin{lemma}[block character-energy identity]
\label{lem:block-character-energy}
Let \(P_1\mathbin{\dot\cup}\cdots\mathbin{\dot\cup}P_s=[n]\), with
\(|P_v|=m_v\), and define
\[
 D_{j_1,\ldots,j_s}=\frac1M
 \bigl|\{(x,y)\in\C^2:d_{P_v}(x,y)=j_v\text{ for all }v\}\bigr|.
\]
For \(0\le w_v\le m_v\), let
\(\Omega_{\boldsymbol w}=\{\alpha\in\F_q^n:
\operatorname{wt}(\alpha|_{P_v})=w_v\text{ for all }v\}\). Then
\begin{equation}
 \sum_{j_1,\ldots,j_s}D_{j_1,\ldots,j_s}
 \prod_{v=1}^s\Kraw_{w_v}(j_v;m_v,q)
 =\frac1M\sum_{\alpha\in\Omega_{\boldsymbol w}}
 \left|\sum_{x\in\C}\chi_\alpha(x)\right|^2\ge0.
\label{eq:block-character-energy}
\end{equation}
\end{lemma}

\begin{proof}
Expand the squared character sums and interchange the sums over \(\alpha\)
and \((x,y)\). On block \(P_v\), summing characters of weight \(w_v\)
gives
\[
 \sum_{\substack{\beta\in\F_q^{P_v}\\\operatorname{wt}(\beta)=w_v}}
 \chi\!\left(\beta\mathbin{\cdot}(x-y)|_{P_v}\right)
 =\Kraw_{w_v}\!\left(d_{P_v}(x,y);m_v,q\right),
\]
since a zero coordinate contributes \(q-1\) and a nonzero coordinate
contributes \(-1\). The blocks are disjoint, so these factors multiply.
\end{proof}

For \(0\le p_1<t_1\), \(0\le p_2<t_2\),
\(0\le\eta\le n-u(t_1,t_2)\), and \(e\in\{0,1\}\), put
\begin{align}
 \Phi_{t_1,t_2;p_1,p_2,\eta,e}
 =\sum_{j,k,c,b}Z_{t_1,t_2;j,k,c,b}\,
 &\Kraw_{p_1}(j;t_1-1,q)
 \Kraw_{p_2}(k;t_2-1,q)\notag\\
 &{}\times\Kraw_\eta(c;n-u(t_1,t_2),q)
 \Kraw_e(b;1,q).
\label{eq:multiblock-spectrum}
\end{align}

\begin{theorem}[four-block positivity]
\label{thm:multiblock-positivity}
For every \(q\)-ary code, not necessarily linear, and every selected two-view
configuration, we have
\[
 \Phi_{t_1,t_2;p_1,p_2,\eta,e}\ge0
\]
on the natural index ranges.
\end{theorem}

\begin{proof}
Fix a recovered coordinate \(i\) with length pair \((t_1,t_2)\). Applied to
\eqref{eq:multiblock-partition}, Lemma~\ref{lem:block-character-energy}
identifies the product-Krawtchouk transform of its four-block distance
distribution with a sum of squared character energies, hence with a
nonnegative quantity. By \eqref{eq:Z-definition},
\(Z_{t_1,t_2;j,k,c,b}\) is \(1/n\) times the sum of these normalized
coordinate distributions over all \(i\) with that ordered length pair.
Therefore \(\Phi_{t_1,t_2;p_1,p_2,\eta,e}\) is the same \(1/n\)-weighted
sum of nonnegative character energies.
\end{proof}

The spectral degree \(\eta\) on \(O_i\) is the feature that distinguishes
the four-block model from a union-only two-view model. Setting \(\eta=0\)
keeps the joint spectrum on the selected union, whereas \(\eta>0\) records
how that local sector extends to a global character.

The same character-energy picture shows exactly how the global spectrum
splits according to whether the character is nonzero at the recovered
coordinate. Define
\begin{align}
 \mathcal S^{(1)}_w
 &=\sum_{t_1,t_2}
   \sum_{p_1+p_2+\eta=w-1}
   \Phi_{t_1,t_2;p_1,p_2,\eta,1},
 &&1\le w\le n,
\label{eq:multiblock-pointed-share}\\
 \mathcal S^{(0)}_w
 &=\sum_{t_1,t_2}
   \sum_{p_1+p_2+\eta=w}
   \Phi_{t_1,t_2;p_1,p_2,\eta,0},
 &&0\le w<n.
\label{eq:multiblock-punctured-share}
\end{align}

\begin{proposition}[global spectral splitting]
\label{prop:multiblock-spectral-shares}
The four-block spectrum satisfies
\begin{equation}
 \mathcal S^{(1)}_w=\frac wnB_w,
 \qquad
 \mathcal S^{(0)}_w=\frac{n-w}{n}B_w
\label{eq:multiblock-spectral-share}
\end{equation}
on the respective ranges.
\end{proposition}

\begin{proof}
For \(\alpha\in\F_q^n\), write
\[
 E_\C(\alpha)=\frac1M\left|\sum_{x\in\C}\chi_\alpha(x)\right|^2.
\]
The global character-energy identity gives
\[
 B_w=\sum_{\operatorname{wt}(\alpha)=w}E_\C(\alpha).
\]
On the other hand, the proof of Theorem~\ref{thm:multiblock-positivity}
shows that, after summing over length pairs and all block degrees of total
weight \(w\), the sector \(e=1\) is
\[
 \frac1n\sum_{i=1}^n
 \sum_{\substack{\operatorname{wt}(\alpha)=w\\ \alpha_i\ne0}}
 E_\C(\alpha).
\]
A character of weight \(w\) occurs in the inner sum for exactly \(w\)
coordinates \(i\), giving \(\mathcal S_w^{(1)}=(w/n)B_w\). The sector
\(e=0\) counts the complementary \(n-w\) coordinates and gives
\(\mathcal S_w^{(0)}=((n-w)/n)B_w\).
\end{proof}

Because every four-block spectral term is nonnegative, the
outside-degree-zero sector retained by the conditioned benchmark is a
sub-sum of \(\mathcal S_w^{(1)}\) or \(\mathcal S_w^{(0)}\). Thus the
exact splitting above yields the character-charge inequalities used in the
comparison projection below.

\subsection{Local collision inequalities}

For a fixed length pair set
\[
 m_{t_1,t_2}=\sum_{j,k,c,b}Z_{t_1,t_2;j,k,c,b}.
\]
Each local projection has minimum distance at least \(\delta\). The
punctured union has the same property.

\begin{lemma}[punctured helper-union distance]
\label{lem:punctured-helper-union}
The set of distinct projections of \(\C\) on
\(S_{i,1}\mathbin{\dot\cup}S_{i,2}\) has minimum distance at least
\(\delta\).
\end{lemma}

\begin{proof}
Take two codewords with distinct projections on the helper union. If their
symbols at \(i\) agree, at least one selected local view differs and
contributes at least \(\delta\) helper differences. If their symbols at
\(i\) differ, both local views differ and each helper block contributes at
least \(\delta-1\) differences. Their disjoint union therefore contributes
at least \(2(\delta-1)\ge\delta\).
\end{proof}

The two complete local views \(T_{i,1}\) and \(T_{i,2}\), together
with the punctured helper union
\(S_{i,1}\mathbin{\dot\cup}S_{i,2}\), give the three collision constraints
used in the LP. The common mechanism is worth making explicit. Suppose a
projection has length \(L\), minimum distance at least \(\delta\), distinct
projected words \(z_1,\ldots,z_N\), and fiber sizes \(f_1,\ldots,f_N\).
Singleton gives \(N\le q^{L-\delta+1}\), while Cauchy--Schwarz gives
\[
 \frac1M\sum_{u=1}^N f_u^2
 \ge \frac{M}{N}
 \ge Mq^{-(L-\delta+1)}.
\]
The left side is exactly the normalized number of ordered codeword pairs
that collide under the projection. Applying this to \(T_{i,1}\),
\(T_{i,2}\), and \(S_{i,1}\mathbin{\dot\cup}S_{i,2}\), and then averaging
coordinates with a fixed length pair, gives
\begin{align}
 q^{t_1-\delta+1}
 \sum_{k,c}Z_{t_1,t_2;0,k,c,0}
 &\ge m_{t_1,t_2},
\label{eq:collision-view-one}\\
 q^{t_2-\delta+1}
 \sum_{j,c}Z_{t_1,t_2;j,0,c,0}
 &\ge m_{t_1,t_2},
\label{eq:collision-view-two}\\
 q^{t_1+t_2-\delta-1}
 \sum_c Z_{t_1,t_2;0,0,c,0}
 &\ge m_{t_1,t_2}.
\label{eq:multiblock-collision}
\end{align}
The last exponent is \((t_1+t_2-2)-\delta+1\), reflecting the length of
the punctured helper union. This is the extra collision information obtained
by keeping the two repair views simultaneously.

\subsection{The LP and projection hierarchy}

The four-block LP is now the intersection of the three layers above. It
maximizes
\[
 M=\sum_{\ell=0}^n A_\ell
\]
over the global Delsarte variables \(A\) and nonnegative four-block
variables \(Z\), subject to the centered and support constraints
\eqref{eq:Z-centered-zero}--\eqref{eq:Z-global-zero}, the positivity
constraints of Theorem~\ref{thm:multiblock-positivity}, and the collision
rows \eqref{eq:collision-view-one}--\eqref{eq:multiblock-collision}. Denote
its optimum by \(U_{\FB}\).

To isolate what is lost when positive outside spectral degrees are
discarded, we define the conditioned coarsening explicitly. For
\((t_1,t_2)\in\mathcal T_2\), retain nonnegative variables
\[
 Y_{\ell;t_1,t_2;j,k,b},
\]
where \(0\le j<t_1\), \(0\le k<t_2\), \(b\in\{0,1\}\), and
\begin{equation}
 j+k+b\le \ell\le j+k+b+n-u(t_1,t_2).
\label{eq:cond-range}
\end{equation}
The variable \(Y\) records the same cell as \(Z\), with the outside distance
replaced by the total distance; for a code-induced point,
\begin{equation}
 Y_{\ell;t_1,t_2;j,k,b}
 =Z_{t_1,t_2;j,k,\ell-j-k-b,b}.
\label{eq:cond-reindex}
\end{equation}
The centered identities become
\begin{align}
 \sum_{t_1,t_2,j,k}Y_{\ell;t_1,t_2;j,k,0}
 &=\frac{n-\ell}{n}A_\ell,
\label{eq:cond-centered-zero}\\
 \sum_{t_1,t_2,j,k}Y_{\ell;t_1,t_2;j,k,1}
 &=\frac{\ell}{n}A_\ell.
\label{eq:cond-centered-one}
\end{align}
We retain the local support zeros
\[
 Y_{\ell;t_1,t_2;j,k,b}=0
 \quad\text{if}\quad
 1\le j+b<\delta\ \text{ or }\ 1\le k+b<\delta.
\]
Together with the global support of \(A\), the centered identities also
eliminate every cell with \(1\le\ell<d_\star\).

The spectral part of this coarsening keeps only the outside-degree-zero
sector. Define
\begin{equation}
 \Psi_{t_1,t_2;p_1,p_2,e}
 =\sum_{\ell,j,k,b}Y_{\ell;t_1,t_2;j,k,b}
 \Kraw_{p_1}(j;t_1-1,q)
 \Kraw_{p_2}(k;t_2-1,q)
 \Kraw_e(b;1,q),
\label{eq:cond-spectrum}
\end{equation}
for the natural degree ranges, and impose
\begin{equation}
 \Psi_{t_1,t_2;p_1,p_2,e}\ge0.
\label{eq:cond-spectrum-positive}
\end{equation}
Thus \(\Psi\) is precisely the \(\eta=0\) slice of the four-block spectrum.
For \(1\le w\le n\) and \(1\le w<n\), respectively, define
\begin{align}
 Q_w^{(1)}
 &=\sum_{t_1,t_2}\sum_{p_1+p_2=w-1}
   \Psi_{t_1,t_2;p_1,p_2,1},
\label{eq:cond-pointed-charge-def}\\
 Q_w^{(0)}
 &=\sum_{t_1,t_2}\sum_{p_1+p_2=w}
   \Psi_{t_1,t_2;p_1,p_2,0},
\label{eq:cond-punctured-charge-def}
\end{align}
and retain the global character-charge bounds
\begin{align}
 Q_w^{(1)}&\le\frac{w}{n}B_w,
 &&1\le w\le n,
\label{eq:cond-pointed-charge}\\
 Q_w^{(0)}&\le\frac{n-w}{n}B_w,
 &&1\le w<n.
\label{eq:cond-punctured-charge}
\end{align}
These inequalities are the part of the full spectral splitting visible when
all positive outside degrees are discarded.

Finally, with
\[
 m^{\Cond}_{t_1,t_2}
 =\sum_{\ell,j,k,b}Y_{\ell;t_1,t_2;j,k,b},
\]
we retain all three collision rows, including the punctured helper-union
constraint:
\begin{align}
 q^{t_1-\delta+1}
 \sum_{\ell,k}Y_{\ell;t_1,t_2;0,k,0}
 &\ge m^{\Cond}_{t_1,t_2},
\label{eq:cond-collision-one}\\
 q^{t_2-\delta+1}
 \sum_{\ell,j}Y_{\ell;t_1,t_2;j,0,0}
 &\ge m^{\Cond}_{t_1,t_2},
\label{eq:cond-collision-two}\\
 q^{t_1+t_2-\delta-1}
 \sum_{\ell}Y_{\ell;t_1,t_2;0,0,0}
 &\ge m^{\Cond}_{t_1,t_2}.
\label{eq:cond-collision-union}
\end{align}
The conditioned LP maximizes the same objective \(M=\sum_\ell A_\ell\) over
\(A\) and \(Y\), subject to the global Delsarte layer and the constraints
above. Denote its optimum by \(U_{\Cond}\). It is used only as a benchmark
for the information carried by positive spectral degrees on \(O_i\).

\begin{theorem}[validity and comparison projections]
\label{thm:multiblock-validity-projections}
Every code with all-symbol local-distance
\((r,\delta,a)\)-availability, \(a\ge2\), induces a feasible four-block
point of objective \(|\C|\). Moreover:

\begin{enumerate}[label=\textup{(\roman*)}]
\item merging either helper block with \(O_i\) gives a feasible
point of the one-view three-block LP of \cite{paperone2026};
\item assigning each cell to \(\ell=j+k+c+b\) gives a feasible point of the
conditioned coarsening.
\end{enumerate}
Consequently,
\begin{equation}
 M_{\max}(q,n,d,r,\delta,a)
 \le U_{\FB}
 \le\min\{U_{\Cond},U_{\mathrm{3blk}}\}.
\label{eq:multiblock-comparison}
\end{equation}
\end{theorem}

\begin{proof}
The preceding identities construct a feasible four-block point from any
code once two views are selected; none of them uses linearity.

For (i), retain one helper distance and replace the outside distance by the
sum of \(c\) and the omitted helper distance. Krawtchouk convolution writes
each transform on this merged block as a sum of nonnegative four-block
transforms, using
\[
 \Kraw_s(x+y;m_1+m_2,q)
 =\sum_{\rho+\eta=s}
 \Kraw_\rho(x;m_1,q)\Kraw_\eta(y;m_2,q).
\]
The centered, support, and collision rows project by summation.

For (ii), use the reindexing \eqref{eq:cond-reindex}. Equation
\eqref{eq:Z-global-distance} gives \eqref{eq:cond-range}, while
\eqref{eq:Z-centered-zero}--\eqref{eq:Z-centered-one} give
\eqref{eq:cond-centered-zero}--\eqref{eq:cond-centered-one}. The local zero
conditions are unchanged. Moreover,
\[
 \Psi_{t_1,t_2;p_1,p_2,e}
 =\Phi_{t_1,t_2;p_1,p_2,0,e},
\]
so \eqref{eq:cond-spectrum-positive} follows from four-block positivity.
All three collision rows project directly, including the punctured
helper-union constraint \eqref{eq:cond-collision-union}.

For the charge bounds, \(Q_w^{(1)}\) and \(Q_w^{(0)}\) are the \(\eta=0\)
sub-sums of \(\mathcal S_w^{(1)}\) and \(\mathcal S_w^{(0)}\), respectively.
All four-block spectral terms are nonnegative, so
Proposition~\ref{prop:multiblock-spectral-shares} gives
\eqref{eq:cond-pointed-charge}--\eqref{eq:cond-punctured-charge}. The
objective is unchanged in both projections, proving
\eqref{eq:multiblock-comparison}.
\end{proof}

The two projections lose different information. The one-view projection
keeps the full outside spectrum but forgets one helper block as a separate
local view. The conditioned projection keeps both helper distances and the
exact total distance but forgets every positive spectral degree on \(O_i\).
Thus the two comparison projections can lose different sources of
strength.

The number of four-block local variables is
\begin{equation}
 N_Z=
 2\!\!\sum_{\substack{(t_1,t_2)\in\mathcal T_2}}
 (n-t_1-t_2+2)t_1t_2
 =O(nr^4)
\label{eq:NZ}
\end{equation}
for fixed \(\delta\). Since \(\ell=j+k+c+b\) is a reindexing of the outside
distance \(c\), the conditioned coarsening has exactly the same number of
local variables.

\section{Exact bounds for availability \texorpdfstring{$a=2$}{a=2}}
\label{sec:exact-bounds}

The theorem-level claims in this section use exact arithmetic throughout.
For every certified parameter set, the three displayed LPs are reconstructed
independently; rational primal and dual solutions are checked constraint by
constraint, and their objectives agree exactly. The attaining construction is
then checked directly against Definition~\ref{def:availability} and against
the induced four-block constraints. The LP upper bounds do not assume
linearity; linearity is used only to describe the attaining constructions
compactly. Floating-point computations are reported separately in
Section~\ref{subsec:numerical-landscape} and are not used as proofs.

\begin{theorem}
\label{thm:exact-availability-two}
For all-symbol local-distance availability \(a=2\),
\[
\begin{aligned}
 M_{\max}(2,8,4,2,2,2)&=8, &
 M_{\max}(3,7,3,2,2,2)&=27,\\
 M_{\max}(4,7,3,2,2,2)&=64.
\end{aligned}
\]
For each parameter set, the four-block LP is strictly stronger than both
comparison relaxations.
\end{theorem}

The certified LP values and matching construction sizes are summarized in
Table~\ref{tab:multiblock-headlines}.

\begin{table}[htbp]
\centering
\caption{Certified availability-two bounds and matching constructions.}
\label{tab:multiblock-headlines}
\small
\setlength{\tabcolsep}{5pt}
\begin{tabular}{@{}lrrrr@{}}
\toprule
\((q,n,d,r,\delta,a)\) & \(U_{\mathrm{3blk}}\) &
\(U_{\Cond}\) & \(U_{\FB}\) & Witness\\
\midrule
\((2,8,4,2,2,2)\) & \(10\) & \(3508/337\) & \(8\) & \(8\)\\
\((3,7,3,2,2,2)\) & \(8991/137\) & \(2817/71\) & \(27\) & \(27\)\\
\((4,7,3,2,2,2)\) & \(4864/23\) & \(112\) & \(64\) & \(64\)\\
\bottomrule
\end{tabular}
\end{table}

\begin{proof}
The exact four-block dual certificates give the three upper bounds
\(8\), \(27\), and \(64\). For the matching lower bounds, in the binary
case take the linear code whose generator matrix has columns
\[
 (001),(001),(010),(010),(100),(101),(110),(111).
\]
For the ternary case use
\[
 (001),(001),(010),(011),(100),(101),(111).
\]
For the quaternary case, regard the seven Fano columns
\[
 (001),(010),(011),(100),(101),(110),(111)
\]
as columns over \(\F_4\). Every Fano line gives a length-three local
projection of distance two, and two distinct lines through a recovered
point give disjoint helper pairs.

Their row spaces have sizes \(8\), \(27\), and \(64\), respectively, and
minimum distances \(4\), \(3\), and \(4\). An independent exact checker
records two disjoint admissible helper sets at every coordinate, forms the
induced \(Z_{t_1,t_2;j,k,c,b}\)-array, and verifies the centered, spectral,
and collision constraints. Thus these constructions attain the certified
four-block upper bounds.
\end{proof}

In particular,
\[
 U_{\FB}<\min\{U_{\Cond},U_{\mathrm{3blk}}\}
\]
in all three rows. Hence the four-block LP is strictly stronger than
either coarsening in these examples, rather than merely a different
parameterization of the same relaxation.

The reproducibility repository contains implementation-independent builders
and checkers for the displayed LPs, exact rational primal--dual certificates
for the certified entries, and independent construction checks for
Theorem~\ref{thm:exact-availability-two}.

\subsection{Broader numerical landscape}
\label{subsec:numerical-landscape}

The exact rows above do not indicate how often the outside spectrum changes
the conditioned relaxation. In the exploratory scan underlying this paper, the numerical four-block
optimum is strictly smaller than the conditioned optimum in all 20 tested
rows. Table~\ref{tab:multiblock-numerical}
shows representative cases across three alphabet sizes, including the two
currently open rows. These values are floating-point HiGHS outputs, not
certified upper bounds. Conditioned values are rounded to four decimal places;
four-block outputs numerically equal to an integer are displayed as that
integer.

\begin{table}[htbp]
\centering
\caption{Representative availability-two computations, with \(a=2\)
throughout. A match is numerical evidence of sharpness, not a theorem.}
\label{tab:multiblock-numerical}
\small
\setlength{\tabcolsep}{4pt}
\begin{tabular}{@{}lrrrl@{}}
\toprule
\((q,n,d,r,\delta)\) & \(U_{\Cond}^{\rm num}\) &
\(U_{\FB}^{\rm num}\) & Recorded size & Status\\
\midrule
\((2,7,2,2,2)\)  & 9.7391   & 8  & 8  & match\\
\((2,9,4,3,2)\)  & 22.1277  & 16 & 16 & match\\
\((2,13,5,2,2)\) & 32.7805  & 32 & --- & open\\
\((2,12,5,3,2)\) & 39.6800  & 32 & --- & open\\
\((3,8,4,2,2)\)  & 41.2105  & 27 & 27 & match\\
\((4,8,4,2,2)\)  & 120.2892 & 64 & 64 & match\\
\bottomrule
\end{tabular}
\end{table}

The complete exploratory scan is stored in the reproducibility repository.
These values are numerical only; promotion to a theorem requires an exact
rational primal--dual certificate and an independently checked construction.

\section{Conclusion}
\label{sec:limitations}

The multi-view construction organizes locality information by retained order
\(s\). At \(s=1\) it is the three-block model of \cite{paperone2026}; at
\(s=2\) it becomes the four-block model developed here; taking \(s=a\)
retains all available repair views. The exact binary, ternary, and quaternary
examples in Theorem~\ref{thm:exact-availability-two} show that, already at
\(s=2\), both sources of information loss can matter: the one-view
projection forgets one helper block as a separate local view, while the
conditioned projection forgets the positive outside spectrum.

The detailed development focuses on \(s=2\) for exposition and
computation. Larger retained order gives a valid
polynomial-size hierarchy for fixed \(s\), but with a rapidly growing local
state space and a larger family of collision constraints. Determining which
of those higher-order constraints are computationally effective is a natural
next step.

The hierarchy remains a two-codeword relaxation. It does not retain
intersections between repair sets attached to different recovered coordinates,
coordinate labels within a helper set, or compatibility among local
configurations at different coordinates. Retaining such data may strengthen
the bound, but it leads to a different and substantially larger optimization
problem.

\backmatter

\subsection*{Funding}
The work of Ming-Hsuan Kang was supported by the National Science and
Technology Council, Taiwan, under Grant NSTC 115-2115-M-A49-010.  The work of
Maosheng Xiong was supported by the Research Grants Council (RGC) of Hong Kong
under Grant No.~16307524.

\section*{Data and code availability}
The reproducibility repository contains the LP builders, exact rational
primal--dual certificates, construction witnesses and independent checkers,
regression tests, and the complete exploratory scans used in this paper.
Theorem-level certificates are separated from floating-point discovery
computations.

\setlength{\bibsep}{2pt plus 0.3ex}


\end{document}